\documentclass[a4paper,12pt]{article}
\usepackage{graphicx}
\usepackage{multirow}
\usepackage{bbm}
\usepackage[bf,small]{caption}
\usepackage{float}
\usepackage{footnote}
\usepackage{amssymb}
\usepackage{pifont}
\usepackage{ifpdf}
\ifpdf
\else
\usepackage{pstcol,pst-fill,pstricks}
\fi
\usepackage{natbib}
\usepackage{mathpazo}
\usepackage{import}
\usepackage{dsfont}
\usepackage{enumerate}
\usepackage{a4}
\usepackage{lscape}
\usepackage{epsfig}
\usepackage[latin1]{inputenc}
\usepackage[OT1]{fontenc}
\usepackage[T1]{fontenc}
\usepackage{color}
\usepackage{geometry}
\usepackage{mdframed}
\usepackage{tikz}
\usepackage{tkz-graph}
\usepackage[bottom,flushmargin]{footmisc}

\makeatletter
\newcommand*\bigcdot{\mathpalette\bigcdot@{.5}}
\newcommand*\bigcdot@[2]{\mathbin{\vcenter{\hbox{\scalebox{#2}{$\m@th#1\bullet$}}}}}
\makeatother
\usepackage{authblk}
\usepackage{setspace}
 \newcommand{\be}{\begin{equation}}
\newcommand{\ee}{\end{equation}}

\usepackage{mathtools}
\usepackage{amsmath}
\usepackage{amsthm}
\usepackage{thmtools}
\usepackage{calrsfs}

\DeclareMathAlphabet{\pazocal}{OMS}{zplm}{m}{n}
\allowdisplaybreaks

\usepackage[hyperfootnotes=false]{hyperref}
 \hypersetup{colorlinks=true,
 linkcolor=blue,
 citecolor=blue,
 filecolor=black,
 urlcolor=black,
 pdfstartview={Fit},
pdfpagemode=UseNone
 }
\usepackage{cleveref}

\declaretheoremstyle[
  headfont=\normalfont\scshape,
  numbered=unless unique,
  bodyfont=\normalfont,
  spaceabove=1em,
  spacebelow=1em,
]{exmpstyle}

\newcommand{\p}{\partial}

\newcommand{\defeq}{\vcentcolon=}

\def\ee{\mathsf{e}}

\newtheorem{prop}{Proposition}
\newtheorem{lemma}{Lemma}
\theoremstyle{definition}

\newtheorem{theorem}{Theorem}
\newtheorem{corollary}{Corollary}[theorem]

\title{Explicit Consumption Functions with Borrowing Constraints: a Continuous Time Approach\thanks{I would like to thank Nghia Phung Trong for excellent research assistance. I would also like to thank Jeanne Commault for their comments and suggestions.}}
\author[1]{Jordan Roulleau-Pasdeloup}
\affil[1]{Department of Economics, National University of Singapore}
\date{\today}

\begin{document}
\begin{titlepage}
\maketitle
\thispagestyle{empty}

\abstract{There is no known explicit global closed form solution for the standard income fluctuation problem with a borrowing constraint and where wealth accumulates with a constant interest rate $r$. Using a continuous time formulation, I derive an explicit global closed form solution for the case $r=0$ using the Lambert W function. For the case $r>0$, I derive an explicit global closed form approximation that is valid for $r\sim 0$. I then use these to derive explicit expressions for the marginal propensity to consume out of assets and permanent income. I show that the cross-derivative between the two is strictly positive: the consumption consumption is supermodular.}
\vspace{.5cm}
\end{titlepage}
\setstretch{1.5}
\setlength{\parskip}{1em}

\section{Introduction}

In their seminal textbook, \cite{Ljungqvist2018recursive} call the standard savings problem the "Common ancestor" of modern macroeconomics. Another popular name for the exact same problem is an "income fluctuation problem". These have been around at least since \cite{Ramsey1928mathematical} and have been studied using a host of different perspectives. Perhaps the most notable one is \cite{Friedman1957theory}'s "Permanent Income Hypothesis", which features stochastic income. It quickly became apparent that, despite the simple structure of the model, it was not possible to solve the model in closed form for standard utility preferences. Some progress was made for non-standard utility function that produced a certainty equivalence result \textemdash see \cite{Hall1978stochastic}, but most of the progress was brought about by computational advances that permitted a numerical solution to these models. In fact, \cite{Zeldes1989optimal} opens his paper by stating that "\textit{No one has derived closed-form solutions for consumption with stochastic labor
income and constant relative risk aversion utility}"

Following the literature, let us define the consumption function as the optimal level of consumption given initial assets. The main finding from \cite{Zeldes1989optimal} is that, instead of the \textit{linear} consumption function that one gets under certainty equivalence, the consumption function without certainty equivalence is \textit{concave}: a marginal increase in assets at the low end of the asset distribution has much more effects on consumption than on the high end of it. This echoes a similar property that obtains in a savings problem with constant income but with a liquidity constraint, such as the one studied in \cite{Schechtman1977some}. Much like the setup studied in \cite{Zeldes1989optimal}, \cite{Deaton1989saving} concludes that "\textit{there is little hope of recovering closed form solutions}". 

Given the unavailability of a closed form solution for the consumption function under stochastic income and/or liquidity constraints, the focus turned to proving theoretically the \textit{properties} of this consumption function. We know from \cite{Schechtman1976income} that the consumption function has to be \textit{increasing} in assets. We know from \cite{Carroll1996concavity} that, for a broad class of utility functions, the consumption function has to be \textit{concave}. We know from \cite{Benhabib2015wealth} that the consumption function is \textit{asymptotically linear}.\footnote{See also \cite{Gouin2023pareto}.}

These early results have since been extended along many dimensions. What has eluded researchers from \cite{Schechtman1976income} onwards is an \textit{explicit closed-form} solution for the consumption function. Using a continuous time formulation of the income fluctuation problem, \cite{Park2006analytical}, \cite{Jolm2018consumption} and \cite{Fischer2024concave} obtained \textit{implicit} expressions for the consumption function.\footnote{While the analysis in \cite{Park2006analytical} is limited to CRRA utility functions, \cite{Jolm2018consumption} considers a wider class of HARA utility functions. See \cite{Merton1975optimum} for the classical reference on the classification of utility functions. \cite{Fischer2024concave} considers a setup with a specific kind of risk - spanned income risk.} More than four decades ago, \cite{Helpman1981optimal} also obtained implicit expressions for the whole path of consumption (which nests the consumption function as a special case) in a version of \cite{Schechtman1977some}'s model where the net interest rate is $r=0$. In this paper, I consider a version of \cite{Helpman1981optimal}'s framework with a generic time period $\Delta$. For $\Delta>0$, the consumption function is piece-wise linear \textemdash this nests the case of $\Delta=1$ which is the original focus in \cite{Helpman1981optimal}. However, as $\Delta\to 0$, the model converges to the continuous time formulation considered in \cite{Park2006analytical}. For the case of $r=0$, I show that one can can use the continuous time version of \cite{Helpman1981optimal}'s method to obtain an \textit{explicit closed-form} solution for the consumption function. 

The solution revolves around the fact that there exists a one-to-one mapping between initial assets and the time it takes to decumulate assets. This map takes the form of a transcendental equation and I show that it can be solved explicitly if we use the Lambert W function\textemdash more specifically, the second branch $W_{-1}(\cdot)$ of the Lambert function, see \cite{Mezo2022lambert} for a textbook treatment. I show next that this solution strategy does not generalize to the case where $r>0$. However, I obtain a global closed form approximation that is valid for $r\sim 0$. This contrasts with the local solutions for low/arbitrarily high assets obtained in \cite{Achdou2022income}. 

The closed form expression for the consumption function explicitly maps both permanent income and assets to current consumption. I leverage this to derive a full characterization of its properties. Let us denote with $c^\ast(a;y)$ the optimal consumption function that takes initial assets $a$ and permanent income $y$ as arguments. I derive an explicit expression for both the Jacobian vector and the Hessian matrix. These imply that the consumption is increasing and concave in both initial assets and permanent income.
In line with the result obtained in \cite{Commault2025Permanent_income} with a Life-Cycle model without borrowing constraints, I show that the MPC out of permanent income is \textit{increasing} in assets. This is somewhat counter-intuitive as high MPCs are usually thought to be a characteristic of poor agents. In addition, given that the consumption function is twice continuously differentiable, it follows that the MPC out of assets is also increasing in permanent income.

\textbf{Related literature}\textemdash 
As previously stated, early results on the income fluctuation problem with liquidity constraints have been extended along many dimensions. Studying the policy function iteration algorithm usually implemented to solve income fluctuation problems through the lenses of operator theory, \cite{Li2014solving} consider a model with unbounded rewards. \cite{Ma2020income} consider a setup where returns on assets, non-financial income and the discount factor are all stochastic. \cite{Ma2021theory} prove in a general setting that homothetic utility functions are asymptotially linear and analytically characterize the asymptotic MPC. \cite{Ma2022asymptotic} further provide conditions on marginal utility for that property to emerge. In a related contribution using a life-cycle model with finite horizon, \cite{commault2022exact} proves that increasing uncertainty increases consumption growth due to a precautionary motive. 

This paper is also related to a series of papers that uses Constant Absolute Risk Aversion (CARA)-type utility functions to solve income fluctuation problems in closed form. The seminal paper in that regard is \cite{Caballero1990consumption}. The setup in this paper has been extended to a \cite{Bewley1983difficulty}-type model in \cite{Wang2003caballero}. In a similar vein, \cite{Toda2017huggett} studied multiple equilibria in \cite{Huggett1993risk}-type economies. This type of preferences has further been used in \cite{Acharya2023optimal} to study a Heterogenous Agents New Keynesian model. 

Finally, this paper is also related to \cite{Lee2025wealth}, where the author uses a clever specification for the stochastic process governing labor income in order to obtain a closed form solution for the optimal consumption path. 

\section{A Standard Income Fluctuation Problem}

In order to make the transition to continuous time, it will be useful to define a time period $\Delta$ that I will make arbitrarily small later on. With this in mind, I assume that there is a single agent choosing sequences $c(t),a(t)\geq 0$ for $t=0,\Delta,2\Delta,\dots$ in order to maximize the following objective:
\begin{align*}
\sum_{t} &\left(\frac{1}{1+\rho\Delta}\right)^{t/\Delta}\Delta u(c(t))   
\end{align*}
where $\rho \in (0,1)$ is the subjective discount factor and $u(\cdot)$ is twice continuously differentiable, strictly increasing and concave on $[0,+\infty)$. I assume a standard CRRA utility function 
\begin{align*}
u(c(t)) \defeq \frac{c(t)^{1-\gamma}}{1-\gamma},    
\end{align*}
where $\gamma>0$ is the coefficient of relative risk aversion. This maximization is subject to the following sequence of constraints:
\begin{align}
\label{eq:constr_Delta}
a(t+\Delta) &= (1+r\Delta)a(t) + \Delta(y-c(t))\\
0&\leq c(t)\leq a(t)
\end{align}
where $a(0)=a>0$ is given and $y\geq 0$ is a constant stream of income. Following the literature, I assmue that the impatience condition $\rho>r$ holds throughout the paper. As $\Delta \rightarrow 0$, the continuous-time representation of this problem is
\begin{align}
\max_{c(t)}\ \int_{0}^{\infty} e^{-\rho t}&u(c(t))dt \label{eq:obj_cont}\\
\text{s.t} \quad \frac{da(t)}{dt} &= ra(t) + y - c(t) \label{eq:constr_cont} \\
0&\leq c(t)\leq a(t) \label{eq:borrowing constraint}
\end{align}

A well-known property of this class of models is that, given an initial asset $a$, it will be optimal for the consumer to deplete her assets after a finite time duration \textemdash see \cite{Park2006analytical} and more recently \cite{Achdou2022income}. 

Now following \cite{Helpman1981optimal}, let us define an implicit sequence $\mu(T)$ for $T=0,\Delta,2\Delta,\dots$ such that
\begin{align}
\mu(T) + \frac{\Delta y-\mu(T-\Delta)}{1 + r\Delta}  &= \left(\frac{1 + r\Delta}{1+\rho\Delta}\right)^{-\frac{T}{\gamma\Delta}}\Delta y 
\label{eq:implicit_mu}
\end{align}
for $T>0$ with $\mu(0)=0$ given. With $\Delta=1$, this replicates \cite{Helpman1981optimal}'s construction for a discrete time interval $t$. This implicit sequence has the following interpretation: if initial assets are such that $\mu(T-\Delta)<a\leq \mu(T)$, then it means that the the consumer completely depletes her income in $T$ periods. After that, she optimally consumes her income $y$ every period.

Using that implicit sequence and focusing on the case where $r=0$, \cite{Helpman1981optimal} shows that the optimal consumption is piece-wise linear and provides an \textit{implicit} expression for it. I give a graphical illustration in Figure \ref{fig:ca_discrete} below, where the values of $\mu(T)$ are indicated with triangles. For reference, I plot the consumption function under the assumption of no borrowing constraints alongside. Given that $r$ is the real interest rate in that framework, I use a low value of $r=0.01$, which is roughly in line with the average real interest rate in the U.S in 2025. For such a low value of the real interest rate, there is a sizable gap between the consumption with and without borrowing constraints. As $a\to\infty$ however, the results in \cite{Benhabib2015wealth} guarantee that the two will converge to the same consumption function which will be linear in initial assets. 
\begin{figure}[ht]
	\centering
	\bigskip
	\caption{Consumption functions in discrete time}\label{fig:AS_AD_US}
	\vspace*{-.5em}
	{\small
\includegraphics[width=1\textwidth]{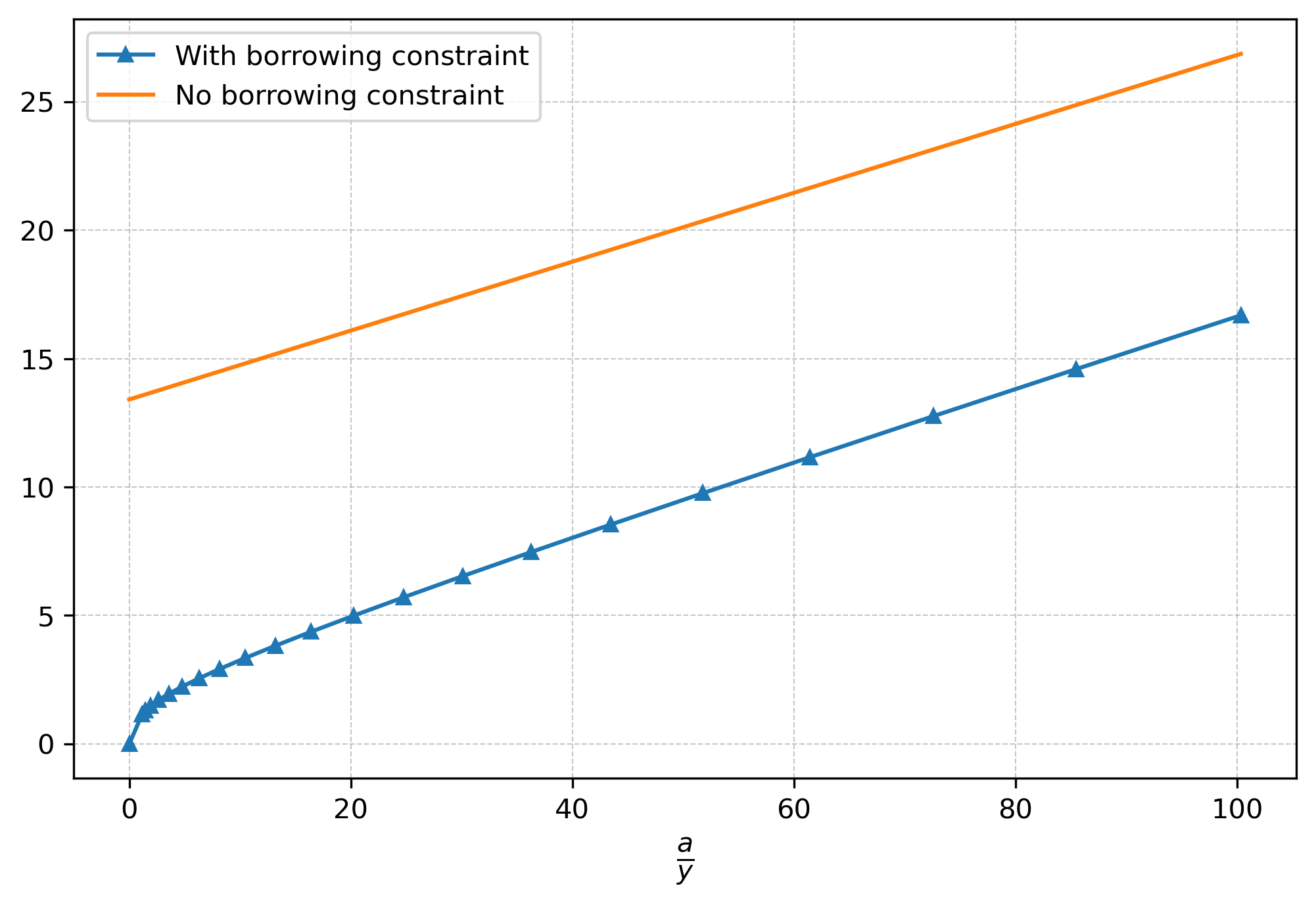}
	}
	\begin{minipage}{0.9\linewidth}
	{\vspace*{-.9em}
 	 \footnotesize\emph{Notes:  We normalize both consumption and initial assets with permanent income. For this figure, I use the following parameter values: $r=0.01, \gamma = 0.5, \rho = 0.08$ and $y=3$.} 
	 }
	\end{minipage}
    \label{fig:ca_discrete}
\end{figure}

Against this backdrop, the main insight is the following: if one makes time continuous by assuming that $\Delta\to 0$, then one can conjecture that the consumption function goes from piecewise linear to a smooth function. I now prove formally that this is the case and characterize this function. Using the functional form for the utility function, I can rewrite equation \eqref{eq:implicit_mu} as follows:
\begin{align*}
\frac{\mu(T) - \mu(T-\Delta)}{\Delta} + r\mu(T) + y &= \left(\frac{1 + r\Delta}{1+\rho\Delta}\right)^{-\frac{T}{\gamma\Delta}} y (1 +r\Delta)\end{align*}
Now considering the limit as $\Delta\to 0$ as well as the properties of the exponential, I get:
\begin{align}
\label{eq:ODE_mu}
\frac{d\mu(T)}{dT} + r\mu(T) + y &= e^{\frac{\rho-r}{\gamma}T}y 
\end{align}
which is a standard Ordinary Differential Equation (ODE), which can be solved using standard methods. For the sake of the argument, let us assume that initial assets are given by $a=\mu(T)$. Then, I have:
\begin{align}
a = \mu(T) = \frac{\gamma y}{r(\gamma-1)+\rho}e^{\frac{(\rho-r) T}{\gamma}} - \frac{y}{r} + e^{-rT} \left( \frac{y}{r} - \frac{\gamma y}{r(\gamma-1)+\rho}  \right)    
\label{eq:a_f(T)}
\end{align}
In that case, remember that it takes exactly $T$ periods for the consumer to run down her initial assets. Under these assumptions, I can characterize the optimal consumption path starting from time $t=0$ as follows:
\begin{theorem}
\label{thm:implicit}
The consumption function given by 
\begin{align}
c^\ast(t)\defeq
\begin{cases}
y\cdot e^{(\rho-r)\frac{T-t}{\gamma}}\quad \text{for}\quad t\leq T\\
y\ \ \hphantom{\cdot e^{(\rho-r)\frac{T-t}{\gamma}}}\quad \text{for}\quad t> T
\end{cases} 
\label{eq:cstar_implicit}
\end{align}
\label{cstar}
maximizes the objective \eqref{eq:obj_cont} subject to the constraint \eqref{eq:constr_cont}.   
\end{theorem}
\begin{proof}
See Appendix \ref{app:Thm1}.    
\end{proof}

At this point, what we have is an expression for consumption as a function of calendar time and the time $T$ that it takes to run down assets. In order to be consistent with the literature, what we really want is to express consumption \textit{explicitly} as a function of \textit{initial assets}. Indeed, the consumption function that is the object of interest in the literature cited in the introduction is essentially equation \eqref{eq:cstar_implicit} with $t=0$: if I start with assets $a$, what is my optimal consumption within the same period?

Given equations \eqref{eq:a_f(T)} and \eqref{eq:cstar_implicit}, in order to compute such an object, what we need to do is to express $T$ as a function of $a$. For future reference, let us denote this function $T:= h(a;y)$, where $y$ appears explicitly as a parameter. Before giving its expression, I show in the following proposition that it exists, is unique, increasing, and concave. 

\begin{prop}
\label{prop:HL}
Assume that the function $\mu$ is given by equation \eqref{eq:a_f(T)} with the assumption that $\rho>r$. Then it follows that
\begin{enumerate}
    \item it is smooth (infinitely differentiable)
    \item it is bijective (one-to-one and onto) and convex from $\mathbb{R}^+$ to $\mathbb{R}^+$
    \item it has an invertible derivative
\end{enumerate}
It follows that, according to the Hadamard-L\'evy theorem, 
\begin{align*}
\tau=\mu^{-1}(a)\defeq h(a;y)    
\end{align*}exists and is unique, strictly increasing and concave in $a$.      
\end{prop}
\begin{proof}
See Appendix \ref{app:proof_HL}.  
\end{proof}
Proposition \ref{prop:HL} guarantees that the function $h(a;y)$ exists and exhibits desirable properties, but is not constructive in the sense that it provides no explicit expression for the function itself. Finding such an expression is difficult because it amounts to solving equation \eqref{eq:a_f(T)}, which is transcendental. At this point, it will be useful to first follow \cite{Helpman1981optimal} and focus on the case when the net real interest rate is given by $r=0$.

\subsection{The special case $r=0$}
\label{sec:r0}

Under the assumption of $r=0$, one can use l'H\^opital's rule to show that the solution of equation \eqref{eq:ODE_mu} is given by:
\begin{align*}
a = \mu(T) &=  y\left(\frac{\gamma}{\rho}e^{\frac{\rho T}{\gamma}}-T-\frac{\gamma}{\rho}\right)\\
&= \frac{y}{b}e^{bT}-yT-\frac{y}{b},
\end{align*}
for all $T\geq 0$ and where $b\defeq \rho/\gamma$. In order to express $T$ as an explicit function of $a$, notice that this condition can be rewritten as follows:
\begin{align}
\frac{y}{b}e^{bT} &= yT + a + \frac{y}{b}\\
&= yT + c
\label{eq:transc_abc}
\end{align}
where the constant $c$ is given by $c = a+y/b$. This is once again a transcendental equation and as a result, it cannot be solved using elementary functions. However, it can be solved using the Lambert $W(\cdot)$ function\textemdash see \cite{Mezo2022lambert} for a comprehensive treatment of this function and its properties. Let us consider the following change of variables $z\defeq -bT - cb/y$. Using this change of variables, equation \eqref{eq:transc_abc} can be rewritten as 
\begin{align}
ze^z=\alpha\Rightarrow z=W_{-1}(\alpha)\ ,\ \alpha = -e^{-\frac{b}{y}\left(a+\frac{\gamma y}{\rho}\right)},  
\label{eq:W_alpha}
\end{align}
where $W_{-1}(\cdot)$ is the second branch of the Lambert $W$ function, which is well defined here given that $\alpha$ takes on values in the interval $(-1/e,0)$ for $a\in (0,+\infty]$. I am now able to state the optimal consumption explicitly as a function of assets. This is described in the following theorem. 

\begin{theorem}
\label{thm:explicit}
The optimal consumption function is given by:
\begin{align}
c^{\ast}(a,t;y) =
\begin{cases}
y\cdot e^{\rho\frac{h(a;y)-t}{\gamma}}\quad\text{for}\quad t\leq T\\    
y\ \ \hphantom{\cdot e^{\rho\frac{h(a;y)-t}{\gamma}}\quad}\text{for}\quad t>T
\end{cases}
\end{align}
where the function $h(a;y)$ is given by
\begin{align*}
h(a;y) &\defeq -\frac{a+\frac{\gamma y}{\rho}}{y}-\frac{\gamma}{\rho} W_{-1}\left(f(a;y)\right)\\
f(a;y) &= -e^{-\frac{b}{y}\left(a+\frac{\gamma y}{\rho}\right)}
\end{align*}
\end{theorem}
\begin{proof}
Given equation \eqref{eq:W_alpha}, we have $z=W_{-1}
(f(a;y))$, where $f(\cdot)$ depends on the structural parameters through the constants $b,c$ and $d$. Next, reverse the change of variables and express $T$ as a function of $z$, which gives a function $T=h(a;y)$. Finally, plug this function for $T$ in \eqref{eq:cstar_implicit}. 
\end{proof}
This immediately gives the following corollary which describes the consumption function at time $t=0$.

\begin{corollary}
\label{cor:explicit}
The consumption function at time $t=0$ is given by
\begin{align*}
c^{\ast}(a;y) \defeq c^{\ast}(a,0;y) = y\cdot e^{\rho\frac{h(a;y)}{\gamma}} ,\end{align*}
where $h(a;y)$ is defined as in Theorem \ref{thm:explicit}. 
\end{corollary}
Note that this is a non-linear expression that is a \textit{global} solution to the income fluctuation problem with a borrowing constraint in continuous time. As a result, this extends the local closed form expressions that have been proposed in the literature. As an example, \cite{Achdou2022income} give expressions for $c^{\ast}(a)$ that are valid for $a\sim 0$ as well as $a\to\infty$. To be fair, these expressions are derived for the general case where $r>0$. The natural question that arises is: can one still provide a global, closed form solution of the income fluctuation problem with a borrowing constraint? I address that issue in the next subsection.

\subsection{The general case: $r>0$}

If one wants to get a corresponding global, closed form solution for the case $r>0$, then one needs to explicitly solve for $T$ as a function of $a$ in equation \eqref{eq:a_f(T)}. However, using the Lambert W function doesn't work in that case. The reason is that, while it can handle an expression that combines a linear and an exponential term, it cannot handle an expression that combines two exponential with different exponents. To the best of my knowledge, inverting a sum of exponentials with different exponents is an open problem.

Given this impossibility, I propose something different here: an expression that is valid for all values of $a$, but for $r\sim 0$. In that sense, it is still global and non-linear but it is now an approximation and not a full-blown solution. The main intuition comes from the fact that :
\begin{align}
a = \mu(T) &=  \frac{\gamma y}{r(\gamma-1)+\rho}e^{\frac{(\rho-r) T}{\gamma}} + \frac{y}{r}(e^{-rT}-1) - \frac{\gamma y}{r(\gamma-1)+\rho}e^{-rT}
\nonumber
\\
&= \frac{\gamma y}{r(\gamma-1)+\rho}e^{\frac{(\rho-r) T}{\gamma}} + \frac{r-\rho}{r(\gamma-1)+\rho}yT-\frac{\gamma y}{r(\gamma-1)+\rho}+o(r),
\label{eq:ar_or2}
\end{align}
where I have used the fact that $e^{-rT} = 1 -rT + o(r)$ for $r\sim 0$. Notice that when $r=0$ this expression exactly nests the one from subsection \ref{sec:r0}. Proceeding in the exact same way as before, one can rewrite equation 
\eqref{eq:ar_or2} as follows:
\begin{align}
\frac{y}{b_r}e^{b_rd_rT} = d_ryT + c_r,
\label{eq:transc_abc_r}
\end{align}
where the coefficients are given by:
\begin{align*}
b_r &= \frac{r(\gamma-1)+\rho}{\gamma } \\ 
c_r &= a+\frac{y}{b_r}+o(r)\\
d_r &= \frac{\rho-r}{r(\gamma-1)+\rho}
\end{align*}
Notice that we have $b_r\to b$ as well as $c_r\to c$ and $d_r\to 1$ as $r\to 0$, where $b,c$ were defined in subsection \ref{sec:r0}. I can now use the same steps as before and use the second branch of the Lambert W function to compute the consumption function in terms of assets. Using the change of variables $z_r\defeq -b_rd_rT-b_rc_r/y$, we can write that
\begin{align*}
z_r = W_{-1}(\alpha_r)\quad,\quad \alpha_r = -e^{-\frac{b_r}{y}(a+\frac{y}{b_r}+o(r))} 
\end{align*}
This naturally leads to the following theorem to characterize the global closed form approximation for the consumption function with $r>0$.
\begin{prop}
\label{prop:gca}
Let the function $h(a;y)$ be defined as in Proposition \ref{prop:HL}. Then the optimal consumption function is given by:
\begin{align*}
c^\ast(a,t;y) =
\begin{cases}
y\cdot e^{\frac{(\rho-r)(h(a;y)-t)}{\gamma}}\quad \text{for}\quad t\leq T\\
y\hphantom{\cdot e^{\frac{(\rho-r)(h(a;y)-t)}{\gamma}}\quad}\ \ \text{for}\quad t> T
\end{cases} 
\end{align*}  where the function $h(a;y)$ is such that
\begin{align*}
h(a;y) &\sim -\frac{1}{d_r\cdot y}\left(a+\frac{y}{b_r}\right)-\frac{1}{b_rd_r} W_{-1}\left(f_r(a,y)\right)\\
f_r(a;y) &= -e^{-\frac{b_r}{y}(a+\frac{y}{b_r})}
\end{align*}
as $r\to 0$. 
\end{prop}

Note that this is only an approximation because the term $o(r)$ contains terms in $T$. In deriving the approximation, I have used the fact that these terms become negligible as $r\sim 0$. In order to visualize how closely the consumption function derived in Proposition \ref{prop:gca} approximates the true consumption function, I solve numerically for the latter using its implicit expression as in \cite{Park2006analytical}. I represent both of them in Figure \ref{fig:imp_exp}. 

\begin{figure}[H]
	\centering
	\bigskip
	\caption{Consumption function in continuous time, small $r$}\label{fig:approx_small_r}
	\vspace*{-.5em}
	{\small
	\includegraphics[width=.9\textwidth]{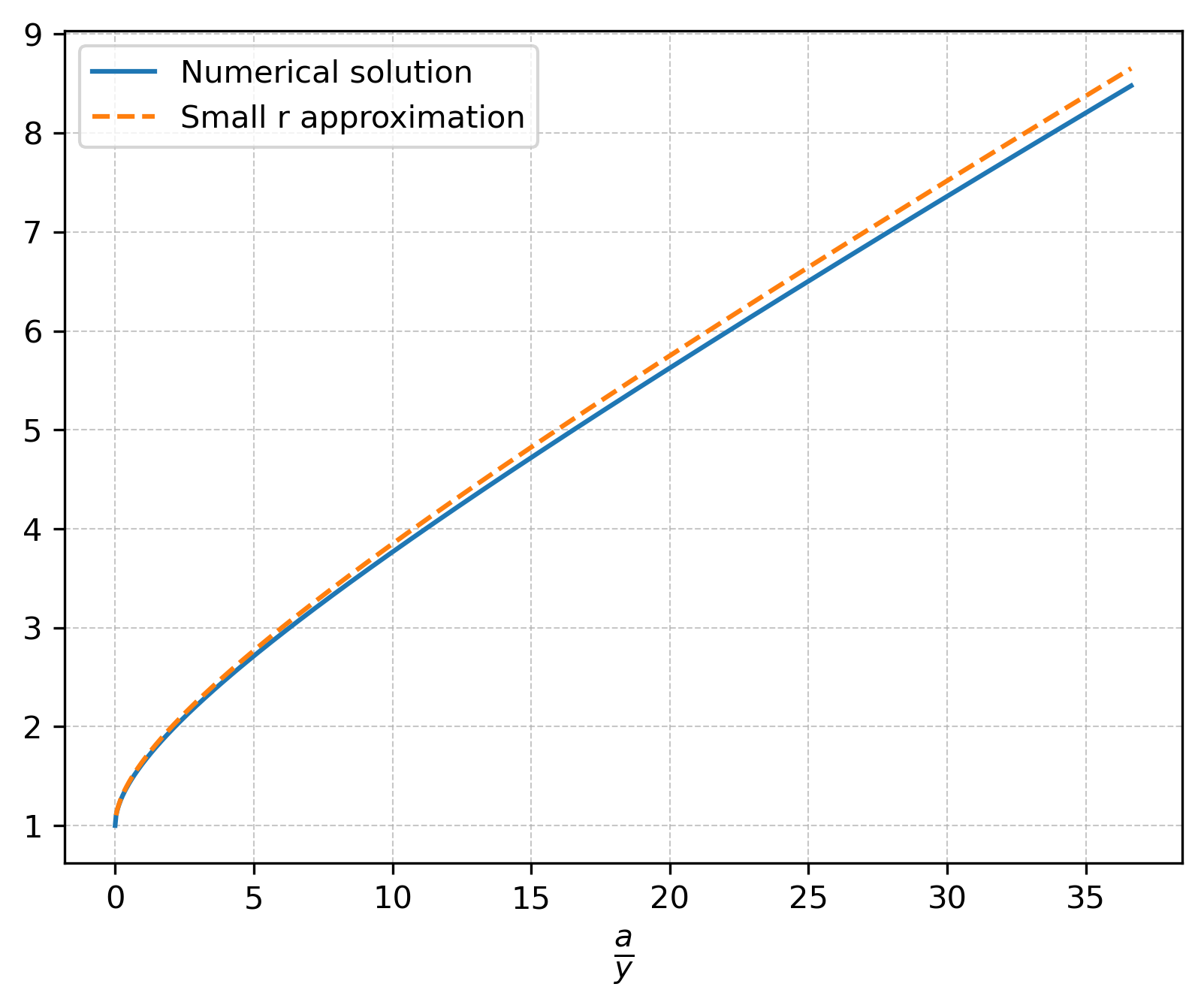}
	}
	\begin{minipage}{0.8\linewidth}
	{\vspace*{-.9em}
 	 \footnotesize\emph{Notes: Both consumption and initial assets are normalized with permanent income. The parameter values are the same as in Figure \ref{fig:ca_discrete}.} 
	 }
	\end{minipage}
    \label{fig:imp_exp}
\end{figure}

The main advantage of the consumption functions derived in Theorem \ref{thm:explicit} and Proposition \ref{prop:gca} is that they give an expression for consumption explicitly in terms of $a,y$ and all the other structural parameters of the model. As a result, it allows one to study in detail how the consumption function changes as a function of these parameters. I leverage that property to provide a complete characterization of the consumption function in terms of assets and permanent income in the next section. 

\section{An Explicit Characterization of the Consumption Function}

While much has been written on whether the consumption function is increasing and/or concave in assets, much less has been written on whether it is increasing and/or concave in permanent income. Recently, \cite{Commault2025Permanent_income} has shown that, in a life-cycle model with a permanent/transitory stochastic process for income, consumption is increasing and concave in permanent income. In this section, I leverage the explicit expression for the consumption function derived in the previous section to provide an explicit characterization of the consumption function and especially its first and second order derivatives. In other words, now that one has an expression for $c^\ast(a;y)$ a natural question is what form do the Jacobian vector and Hessian matrix look like. 

For ease of exposition, I focus my attention on the case with zero interest rate. Note that the consumption functions derived in Theorem \ref{thm:explicit} and Proposition \ref{prop:gca} are isomorphic. It follows that the proofs presented in this section extend to the global closed form approximation derived in Proposition \ref{prop:gca}. For the particular case of $r=0$, I show in the appendix that the following proposition holds:

\begin{prop}
\label{prop:Jacobian}
Assume that the consumption function $c^\ast(a;y)$ is defined as in Corollary \ref{cor:explicit}. Then it follows that its Jacobian vector $\mathbf{J}_c$ is given by:
\begin{align*}
\mathbf{J}_c &\defeq 
\begin{bmatrix}
\frac{\p c^\ast(a;y)}{\p a} & \frac{\p c^\ast(a;y)}{\p y}    
\end{bmatrix}\\
&\ = 
\begin{bmatrix}
\frac{\rho}{\gamma}\frac{w}{1+w} &   -w\left(1+\frac{\rho a}{\gamma y}\frac{1}{1+w}\right)   
\end{bmatrix}
\end{align*}
where I have defined $w\defeq W_{-1}(f(a;y))$ for convenience.
\end{prop}
\begin{proof}
See Appendix \ref{app:proof_Jacobian}.   
\end{proof}
This Jacobian vector gives us the direction in which consumption goes if either initial assets or permanent income increase. The key question now is to find out whether the entries of $\mathbf{J}_c$ are positive or negative. The answer is given in the following Corollary:

\begin{corollary}
\label{cor:Jacobian}
Assume that the impatience condition $(\rho>r)$ holds. Then it follows that both entries of $\mathbf{J}_c$ are strictly positive.     
\end{corollary}
\begin{proof}
See Appendix \ref{app:proof_Jacobian}.   
\end{proof}

The fact that the first entry is strictly positive is an established result in the literature: the consumption function is increasing in the level of initial assets. The fact that the Second entry is strictly positive is intuitive but difficult to prove formally as permanent income usually features either as a stochastic process or as a parameter when it is assumed to be constant \textemdash like in the current paper. In fact, \cite{Achdou2022income} provide some suggestive evidence that consumption varies one for one with $y$ when $a\sim 0$, but they do not prove that consumption increases with income for every $a$. \cite{Commault2025Permanent_income} does prove that the consumption function is increasing in income, but in a setup without borrowing constraints. 

Another feature of the consumption function that has been the subject of a lot of attention in the literature has to do with the second derivative of the consumption function: concavity. Indeed, one robust finding from the existing literature is that stochastic income and borrowing constraints are both responsible for this shape of the consumption function. With the exception of \cite{Commault2025Permanent_income} for a model with stochastic income but without borrowing constraints, there is to the best of my knowledge no known proof that the consumption function is concave in permanent income. Further, second derivatives can also inform us on whether the marginal propensity to consume out of a windfall in income is increasing or decreasing in initial assets $a$. Again, \cite{Commault2025Permanent_income} has shown that in a setup with stochastic income but without borrowing constraints that this MPC is increasing in $a$. In the current setup, if that property is true it implies that the MPC out of a windfall in assets $a$ is increasing in permanent income $y$. This is is due to the fact that the consumption function $c^\ast(a;y)$ has continuous second derivatives. The appropriate object to think about second-order derivatives is the Hessian matrix $\mathbf{H}_c$. The following proposition offers a characterization of that Hessian matrix.

\begin{prop}
\label{prop:Hessian}
Assume that the consumption function is defined as in Corollary \ref{cor:explicit}. Then it follows that:
\begin{align*}
\mathbf{H}_c &\defeq 
\begin{bmatrix}
\frac{\p^2 c^\ast(a;y)}{\p a^2} & \frac{\p^2 c^\ast(a;y)}{\p a \p y} \\
\frac{\p^2 c^\ast(a;y)}{\p y \p a} & \frac{\p^2 c^\ast(a;y)}{\p y^2}
\end{bmatrix}\\
&\ =
\begin{bmatrix}
-\frac{\rho^2}{\gamma^2 y}\frac{w}{(1+w)^3} & \frac{a\rho^2}{\gamma^2 y^2}\frac{w}{(1+w)^3}\\
\frac{a\rho^2}{\gamma^2 y^2}\frac{w}{(1+w)^3} & -\frac{\rho^2a^2}{\gamma^2 y^3}\frac{w}{(1+w)^3}
\end{bmatrix}
\end{align*}
where, as before, I have defined $w\defeq W_{-1}(f(a;y))$ for convenience.
\end{prop}
\begin{proof}
See Appendix \ref{app:proof_Hessian}. 
\end{proof}

The main take-away from Proposition \ref{prop:Hessian} is that all second-order derivatives are proportional to $w/(1+w)^3$. Therefore, one can use the fact that the second branch of the Lambert W function is strictly less than $-1$ on $(-1/e,0)$ to show that this fraction is unambiguously strictly positive. This immediately implies the following Corollary:

\begin{corollary}
\label{cor:Hessian}
Assume that the impatience condition $(\rho>r)$ holds. Then it follows that the diagonal elements of $\mathbf{H}_c$ are strictly negative and that the off-diagonals of $\mathbf{H}_c$ are strictly positive. As a result, $c^\ast(a;y)$ is supermodular.  
\end{corollary}
\begin{proof}
See Appendix \ref{app:proof_Hessian}. 
\end{proof}

To the best of my knowledge, the result that the consumption function is supermodular is new to the literature. There is a whole literature about supermodularity in economic theory, but it revovles around assuming that the utility function is supermodular \textemdash see \cite{Amir2005supermodularity} for a survey. In contrast, the result in the current paper is that the consumption function, which is an equilibrium object, is supermodular. In words, the fact that $\mathbf{H}_c$ is supermodular means that assets and permanent income are complements. This means that, as in \cite{Commault2025Permanent_income}, the MPC out of permanent income is increasing in initial assets $a$.

The intuition as to why the MPC out of permanent income is increasing in assets is as follows. When assets are low, an increase in permanent income is not met with a large increase in consumption due to the risk to hit the borrowing constraint. When assets are higher however, the consumer can stand to benefit from the increase in permanent income for two main reasons. First, if assets are relatively high the consumer has a lot of buffer savings. As a result they understand that the risk to hit the constraint is relatively low. In addition, an increase in permanent income also serve as a signal that all future periods will feature a higher income. As a result, the consumer wants to smooth consumption and bring some of this higher income forward in time. Having high assets enables the consumer to increase her consumption to reflect that. 

Given that the consumption function is twice-differentiable, this also means that the MPC out of assets is increasing in permanent income: a windfall increase in assets will be met with a comparatively higher consumption response for an individual with a higher permanent income. The intuition is similar to the one presented in the last paragraph. This is a property that could be tested using empirical data on consumption.

To conclude, note that the strictly negative diagonal elements of $\mathbf{H}_c$ encode the fact that the consumption function is guaranteed to be concave in both assets and permanent income. Intuitively, an increase in both has a diminishing marginal effect on consumption. While this has been known for a while for initial assets, the diminishing marginal effect of a windfall increase in permanent income for a given level of initial assets hasn't been proved in a setting with borrowing constraints.  

\section{Conclusion}

Given the unavailability of an explicit closed form solution for the consumption function in typical income fluctuation problems with borrowing constraints, one has to compute this object numerically. The same goes for the MPC out of assets that is implied by these models. This has been the case both in discrete as well as continuous time formulations of these models and numerical methods have been developed for both.

I have derived a closed form solution that is both global and explicit for $r=0$. For $r>0$, I have derived a closed form approximation that is valid as long as $r\sim 0$. Using these solutions, I have shown that the MPC out of permanent income is counter-intuitively increasing in the level of initial assets. Further, this enabled me to prove that the consumption function is increasing, concave in permanent income and initial assets, which act as complements.

In addition, the solutions developed in this paper could be used as a benchmark for numerical methods that have been cast in continuous time and assess their reliability when computing objects such as the MPC out of assets.

\bibliographystyle{apalike2}
\bibliography{IFP}

\appendix

\section{Proof of Theorem \ref{thm:implicit}}
\label{app:Thm1}
In order to prove theorem \ref{thm:implicit}, we will need a few intermediate lemmas. 

First, rearrange the constraint \eqref{eq:constr_cont} and multiply both sides by the integrating factor $e^{-rt}$, we get 
\begin{align*}
e^{-rt}c(t) &= e^{-rt}ra(t) - e^{-rt}\frac{da(t)}{dt}  + e^{-rt}y 
\end{align*}
Integrate both sides from 0 to t, we get
\begin{align}
\int_{0}^{t}{e^{-r\tau}c(\tau)}d\tau = -e^{-rt}a(t) + a - \frac{y}{r}(e^{-rt}-1)
\label{eq:budget_identity}
\end{align}
Under the borrowing constraint $a(t)\ge 0$, this yields
\begin{align*}
\int_{0}^{t}{e^{-r\tau}c(\tau)}d\tau \leq a - \frac{y}{r}(e^{-rt}-1)
\end{align*}
which is a necessary and sufficient condition for a consumption plan to be feasible. 

It follows that a consumption plan is feasible if, and only if, it satisfies the following two conditions:
\begin{align}
\label{eq:feasible_1_app}
\int_{0}^{t}{e^{-r\tau}c(\tau)}d\tau &\leq a - \frac{y}{r}(e^{-rt}-1) \\
c(t)&\geq 0
\label{eq:feasible_2_app}
\end{align}
We can now state the first lemma.
\begin{lemma}[Finiteness of maximal value] 
The maximal value of \eqref{eq:obj_cont} subject to \eqref{eq:constr_cont} and \eqref{eq:borrowing constraint} is finite. Specifically, for any feasible plan $c(t)$ 
\begin{align*}
\int_{0}^\infty e^{-\rho t} u(c(t))dt \leq \dfrac{u\left(\rho a + y\right)}{\rho}
\end{align*}
\label{lm:lemma1}
\end{lemma}
\begin{proof}
If $c(t)$ is feasible, from \eqref{eq:feasible_1_app} we have
\begin{align}
\int_{0}^{\infty} e^{(r-\rho)t} \int_{0}^{t}{e^{-r\tau}c(\tau)} d\tau dt \leq \int_{0}^{\infty} e^{(r-\rho)t} \left[ a - \frac{y}{r}(e^{-rt}-1)\right] dt
\label{ieq:lemma1}
\end{align}
The RHS of \eqref{ieq:lemma1} evaluates to  
\begin{align*}
\frac{a}{\rho-r} + \frac{y}{\rho(\rho-r)} 
\end{align*}
Let us define 
\begin{align*}
C(t) \defeq \int_{0}^t e^{-r\tau} c(\tau)d\tau
\end{align*}
We set $u=C(t)$ as well as $dv=e^{(r-\rho) t}dt$. We then have $du = e^{-rt}c(t)dt$ and $v=\dfrac{e^{(r-\rho)t}}{r-\rho}$. Using integration by parts, the LHS of \eqref{ieq:lemma1} becomes
\begin{align}
\frac{e^{(r-\rho)t}}{r-\rho}C(t) \Bigg|_{0}^{\infty} + \frac{1}{\rho-r} \int_0^\infty e^{-\rho t}c(t)dt
\label{eq:LHS_lemma1}
\end{align}
Note that from \eqref{eq:budget_identity}, 
\begin{align*}
\int_{0}^{t}{e^{-r\tau}c(\tau)}d\tau +e^{-rt}a(t) &=  a - \frac{y}{r}(e^{-rt}-1) \\
\Rightarrow  \int_{0}^{\infty}{e^{-r\tau}c(\tau)}d\tau + \lim_{t\to\infty} e^{-rt}a(t) &=  a + \frac{y}{r} \\
\Rightarrow \lim_{t\to\infty} C(t) = \int_{0}^{\infty}{e^{-r\tau}c(\tau)}d\tau  &\leq  a + \frac{y}{r}
\end{align*}
where the last line is due to $a(t)\geq0$ and $c(t)\geq0$. Therefore, the first term of \eqref{eq:LHS_lemma1} evaluates to 0. Then, inequality \eqref{ieq:lemma1} simplifies to 
\begin{align}
\frac{1}{\rho-r} \int_0^\infty e^{-\rho t}c(t)dt &\leq \frac{a}{\rho-r} + \frac{y}{\rho(\rho-r)} \nonumber \\
\int_0^\infty e^{-\rho t}c(t)dt &\leq a + \frac{y}{\rho}
\label{ieq:constraint_lemma1}
\end{align}
That is, any feasible path is subject to \eqref{ieq:constraint_lemma1}. Now, consider the problem 
\begin{align*}
\max_{c(t)}\ \int_{0}^{\infty} e^{-\rho t}&u(c(t))dt 
\end{align*}
subject to \eqref{ieq:constraint_lemma1}. It can be shown that the solution to this problem is $c(t) = \overline{c} = \rho a + y$, and the maximal value is $\dfrac{u\left(\rho a + y\right)}{\rho}$. Thus, the present discounted utility of any feasible path is bounded from above by $\dfrac{u\left(\rho a + y\right)}{\rho}$, which completes the proof.
\end{proof}

We have now established that if a consumption plan is feasible, then the integral of discounted utility is bounded above. We now need to establish that the consumption plan $c^{\ast}(t)$ is indeed feasible. We prove that this is the case in the following lemma. 

The next result shows that the consumption plan $c^{\ast}(t)$ is indeed feasible.
\begin{lemma}[Feasibility of $c^{\ast}(t)$] The consumption plan $c^{\ast}(t)$ defined in \eqref{eq:cstar_implicit} is feasible.
\label{lm:lemma2}
\end{lemma}
\begin{proof}
Let us reproduce equation \eqref{eq:ODE_mu} for convenience:
\begin{align}
\label{eq:ODE_mu_app}
\frac{d\mu(x)}{dx} + r\mu(x) + y &= e^{\frac{\rho-r}{\gamma}x}y \\
\nonumber\Rightarrow\quad e^{rx}\frac{d\mu(x)}{dx} + e^{rx}r\mu(x) + e^{rx}y &= e^{rx}e^{\frac{\rho-r}{\gamma}x}y
\end{align}
Integrate both sides from 0 to $T$, we get
\begin{align*}
\int_{0}^{T} e^{rx}\frac{d\mu(x)}{dx} + e^{rx}r\mu(x) + e^{rx}y ~dx &= \int_{0}^{T}e^{rx}e^{\frac{\rho-r}{\gamma}x}y~dx\\
\Rightarrow\quad e^{rT}\mu(T) - \mu(0) +\frac{y}{r}(e^{rT}-1)  &= \int_{0}^{T}e^{rx}e^{\frac{\rho-r}{\gamma}x}y~dx\\
\Rightarrow\quad \mu(T) + \frac{y}{r}(1-e^{-rT})  &= \int_{0}^{T}e^{r(x-T)}e^{\frac{\rho-r}{\gamma}x}y~dx
\end{align*}
where we have used the fact that $\mu(0) = 0$ in the last equality. Finally, apply the change of variable $t = T-x$ on the right-hand side, we get 
 \begin{align}
 \mu(T) + \frac{y}{r}(1-e^{-rT})  &= \int_{0}^{T}e^{-rt}e^{\frac{(\rho-r)(T-t)}{\gamma}x}y~dt \\
 &= \int_{0}^{T}{e^{-rt}c^{\ast}(t)}~dt
 \label{eq:cont_mu}
 \end{align}
where the last line follows from the definition of $c^{\ast}(t)$. 

Suppose $a = \mu(T)$. For $t<T$, 
\begin{align*}
    \int_{0}^{t}{e^{-r\tau}c^{\ast}(\tau)}d\tau &= -e^{-rt}a(t) + \mu(T) - \frac{y}{r}(e^{-rt}-1) \\
    &= -e^{-rt}a(t) + \mu(T) - \frac{y}{r}(e^{-rt}-e^{-rT}+e^{-rT}-1) \\
    &= -e^{-rt} \left[  a(t) - \frac{y}{r}(e^{-r(T-t)}-1) \right] + \mu(T) - \frac{y}{r}(e^{-rT}-1) \\
    &= -e^{-rt} \left[  a(t) - \frac{y}{r}(e^{-r(T-t)}-1) \right] + \int_{0}^{T}{e^{-r\tau}c^{\ast}(\tau)}d\tau
\end{align*}
where the last equation follows from \eqref{eq:cont_mu}. Rearrange this to arrive at 
\begin{align*}
e^{-rt} \left[  a(t) - \frac{y}{r}(e^{-r(T-t)}-1) \right] &= \int_{t}^{T}{e^{-r\tau}c^{\ast}(\tau)}d\tau \\
&= \int_{t}^{T} {e^{-r\tau}} e^{\frac{(\rho-r)(T-t)}{\gamma}\tau}yd\tau \\ 
&= \int_{0}^{T-t} {e^{-r(t+x)}}e^{\frac{(\rho-r)(T-t-x)}{\gamma}}ydx
\end{align*}
where the change of variable $x \defeq \tau - t$ was applied in the third equality. This gives us
\begin{align*}
a(t) - \frac{y}{r}(e^{-r(T-t)}-1)  &=  \int_{0}^{T-t} {e^{-rx}} e^{\frac{(\rho-r)(T-t-x)}{\gamma}}ydx \\
&=  \mu(T-t) - \frac{y}{r}(e^{-r(T-t)}-1) \\
\Rightarrow a(t) &= \mu(T-t) > 0
\end{align*}
This means the borrowing constraint holds strictly for $t<T$. Meanwhile, for $t \geq T$, 
\begin{align*}
    \int_{0}^{t}{e^{-r\tau}c^{\ast}(\tau)}d\tau &= \int_{0}^{T}{e^{-r\tau}c^{\ast}(\tau)}d\tau + \int_{T}^{t}{e^{-r\tau}c^{\ast}(\tau)}d\tau \\
     &= \int_{0}^{T}{e^{-r\tau} e^\frac{(\rho-r)(T-\tau)}{\gamma}}d\tau + \int_{T}^{t}{e^{-r\tau}y}d\tau \\
     &= \mu(T) - \frac{y}{r}(e^{-rT}-1) - \frac{y}{r}(e^{-rt}-e^{-rT}) \\
     &= \mu(T) - \frac{y}{r}(e^{-rt}-1) 
\end{align*}
and so condition \eqref{eq:feasible_1_app} holds with equality. This completes the proof. 
\end{proof}

The final step is to prove that the consumption plan $c^{\ast}(t)$ is optimal. Lemma \ref{lm:lemma2} tells us that $c^{\ast}(t)$ is feasible. By Lemma \ref{lm:lemma1}, we can define the difference 
\begin{align*}
\delta \defeq \int_{0}^{\infty} e^{-\rho t}\left[u(c(t))-u(c^{\ast}(t))\right]dt
\end{align*}
for any feasible consumption path $c(t)$. Since $u(\cdot)$ is strictly concave,
\begin{align*}
\delta &< \int_{0}^{\infty} e^{-\rho t}u'(c^{\ast}(t))(c(t)-c^{\ast}(t))dt \\
&= \int_{0}^{T} e^{-\rho t}u'(c^{\ast}(t))(c(t)-c^{\ast}(t))dt + \int_{T}^{\infty} e^{-\rho t}u'(c^{\ast}(t))(c(t)-c^{\ast}(t))dt
\end{align*}
By definition of $c^{\ast}(t)$, the first term of the RHS is
\begin{align} 
&\int_{0}^{T} e^{-\rho t}u'(c^{\ast}(t))[c(t)-c^{\ast}(t)]dt \nonumber\\
&=\int_{0}^{T} e^{-\rho t} e^{(\rho-r)(t-T)}u'(y) [c(t)-c^{\ast}(t)]dt \nonumber\\
&= e^{(r-\rho)T}u'(y) \int_{0}^{T} e^{-rt}[c(t)-c^{\ast}(t)]dt
\label{eq:theorem1_RHS1}
\end{align}
Meanwhile, the second term of the RHS is  
\begin{align}
& \int_{T}^{\infty} e^{-\rho t}u'(c^{\ast}(t))(c(t)-c^{\ast}(t))dt \nonumber\\
& = \int_{T}^{\infty} e^{-\rho t}u'(y)(c(t)-y)dt \nonumber\\
& = \int_{T}^{\infty} e^{(r-\rho)t}e^{-rt}u'(y)(c(t)-y)dt \nonumber\\
& = \int_{0}^{\infty} e^{(r-\rho)(T+t)}e^{-r(t+T)}u'(y)[c(T+t)-y]dt \nonumber\\
& = e^{(r-\rho)T}u'(y)\int_{0}^{\infty} e^{-\rho t - rT}[c(T+t)-y]dt
\label{eq:theorem1_RHS2}
\end{align}
Note that since c(t) is a feasible consumption path, condition \eqref{eq:feasible_1_app} implies
\begin{align}
\int_{0}^{\infty} e^{(r-\rho)t}\int_{0}^{t+T} e^{-r\tau}c(\tau)d\tau~dt &\leq \int_{0}^{\infty} e^{(r-\rho)t} \left[ a -\frac{y}{r}(e^{-r(T+t)} -1) \right] dt \nonumber\\
&= \int_{0}^{\infty} ae^{(r-\rho)t} -\frac{y}{r}e^{-\rho t - rT} + \frac{y}{r}e^{(r-\rho)t} ~dt \nonumber\\
& = \frac{a}{\rho-r} - \frac{ye^{-rT}}{r\rho} + \frac{y}{r(\rho-r)}
\label{ieq:theorem1_doublesum}
\end{align}
However, 
\begin{align*}
&\int_{0}^{\infty} e^{(r-\rho)t}\int_{0}^{t+T} e^{-r\tau}c(\tau)d\tau~dt\\
&= \int_{0}^{\infty} e^{(r-\rho)t} \left( \int_{0}^{T} e^{-r\tau}c(\tau)d\tau + \int_{T}^{t+T} e^{-r\tau}c(\tau)d\tau \right)~dt\\ 
&=\int_{0}^{\infty} e^{(r-\rho)t}\int_{0}^{T} e^{-r\tau}c(\tau)d\tau dt + \int_{0}^{\infty} e^{(r-\rho)t}\int_{0}^{t+T} e^{-r\tau}c(\tau)d\tau dt
\end{align*}
The first term on the RHS evaluates to 
\begin{align*}
\frac{1}{\rho-r}\int_{0}^{T} e^{-r\tau}c(\tau)d\tau 
\end{align*}
The second term on the RHS, by changing the order of integration, gives
\begin{align*}
\int_{T}^{\infty} \int_{\tau-T}^{\infty} e^{(r-\rho)t} e^{-r\tau}c(\tau) dt d\tau &= \int_{T}^{\infty} \frac{e^{(r-\rho)(\tau-T)}}{\rho-r} e^{-r\tau}c(\tau) d\tau \\
&= \int_{0}^{\infty} \frac{e^{(r-\rho)t}}{\rho-r} e^{-r(t+T)}c(t+T) dt \\
&= \frac{1}{\rho-r}\int_{0}^{\infty} e^{-\rho t - rT} c(t+T) dt
\end{align*}
where a change of variable is applied in the second equality. Altogether, inequality \eqref{ieq:theorem1_doublesum} becomes 
\begin{align*}
&\frac{1}{\rho-r}\int_{0}^{T} e^{-r\tau}c(\tau)d\tau + \frac{1}{\rho-r}\int_{0}^{\infty} e^{-\rho t - rT} c(t+T) dt\\
&= \int_{0}^{\infty} e^{(r-\rho)t}\int_{0}^{t+T} e^{-r\tau}c(\tau)d\tau~dt
\leq \frac{a}{\rho-r} - \frac{ye^{-rT}}{r\rho} + \frac{y}{r(\rho-r)}
\end{align*}
and so
\begin{align}
\int_{0}^{T} e^{-rt}c(t)dt + \int_{0}^{\infty} e^{-\rho t - rT} c(t+T) dt
&\leq a - \frac{(\rho-r)ye^{-rT}}{r\rho} + \frac{y}{r} \nonumber\\
&= a + \frac{ye^{-rT}}{\rho} - \frac{ye^{-rT}}{r} + \frac{y}{r} \nonumber\\
\Rightarrow \int_{0}^{\infty} e^{-\rho t - rT} c(t+T) dt - \frac{ye^{-rT}}{\rho} &\leq a - \frac{ye^{-rT}}{r} + \frac{y}{r} -\int_{0}^{T} e^{-rt}c(t)dt \nonumber\\
\Rightarrow \int_{0}^{\infty} e^{-\rho t - rT} [c(t+T)-y] dt  &\leq a - \frac{ye^{-rT}}{r} + \frac{y}{r} -\int_{0}^{T} e^{-rt}c(t)dt \label{ieq:theorem1_RHS2}
\end{align}

We now return to the difference $\delta$ between the present discount value of $c^{\ast}(t)$ and that of another feasible $c(t)$. Equalities \eqref{eq:theorem1_RHS1} and \eqref{eq:theorem1_RHS2} yield
\begin{align*}
\delta &< e^{(r-\rho)T}u'(y) \int_{0}^{T} e^{-rt}[c(t)-c^{\ast}(t)]dt +  e^{(r-\rho)T}u'(y)\int_{0}^{\infty} e^{-\rho t - rT}(c(T+t)-y)dt \\
&\leq e^{(r-\rho)T}u'(y) \int_{0}^{T} e^{-rt}[c(t)-c^{\ast}(t)]dt + e^{(r-\rho)T}u'(y) \left[ a - \frac{ye^{-rT}}{r} + \frac{y}{r} -\int_{0}^{T} e^{-rt}c(t)dt  \right]\\
&=e^{(r-\rho)T}u'(y) \left[ -\int_{0}^{T} e^{-rt}c^{\ast}(t)dt + a - \frac{ye^{-rT}}{r} + \frac{y}{r}  \right]\\
&=0
\end{align*}
where the second line holds due to inequality \eqref{ieq:theorem1_RHS2}. Therefore, $c^{\ast}(t)$ is the unique feasible consumption plan that maximizes the present discount value of utility. This completes the proof for Theorem \ref{thm:implicit}.

\section{Proof of Proposition \ref{prop:HL}}
\label{app:proof_HL}

For part 1, note that function $\mu(T)$ is a linear combination of exponentials and a constant function, both of which are infinitely differentiable on $\mathbb{R}^+$. As a result, it follows that $\mu(T)$ is infinitely differentiable on $\mathbb{R}^+$.

For part 2, I start by proving that $\mu(T)$ is injective using the fact that it is strictly monotone. First of all, note that the derivative of $\mu(T)$ is given by:
\begin{align*}
\mu^{'}(T) = y\frac{\rho - r}{r(\gamma-1)+\rho}\left(e^{\frac{\rho-r}{\gamma}T} - e^{-rT}\right)    
\end{align*}
Notice that under the assumptions that $\rho>r$ and $\gamma>0$ one can write
\begin{align*}
r(\gamma-1)>-r\quad\Rightarrow\quad \rho + r(\gamma-1) > \rho-r>0.    
\end{align*}
As a result, the constant term in front of the parentheses is strictly positive. Given that $\rho-r>0$ and $\gamma>0$, it follows that $e^{\frac{\rho-r}{\gamma}T}>1$ for $T\in \mathbb{R}^+$. Further, $r>0$ guarantees that $e^{-rT}<1$ for $T\in \mathbb{R}^+$. Putting it all together, we get 
\begin{align*}
\mu^{'}(T) = y\frac{\rho - r}{r(\gamma-1)+\rho}\left(e^{\frac{\rho-r}{\gamma}T} - e^{-rT}\right) >0\quad \text{for}\quad T\in \mathbb{R}^+.   
\end{align*}
As a result, $\mu(T)$ is strictly monotonic and thus injective. 

To prove surjectivity, it needs to be shown that $\mu(T)$ covers $\mathbb{R}^+$. To establish it, note that by definition we have $\mu(0)=0$. Further, given that $\rho>r$ the first exponential term of $\mu(T)$ dominates and we have $\lim_{T\to\infty} \mu(T) = +\infty$. As a result, given that $\mu(T)$ is continuous and increasing it covers $\mathbb{R}^+$. Putting it all together, this establishes that $\mu(T)$ is bijective on $\mathbb{R}^+$. 

For part 3, it needs to be shown that $\mu^{'}(T)$ is bijective and thus invertible. To establish injectivity, note that both $e^{\frac{\rho-r}{\gamma}T}$ and $-e^{-rT}$ are strictly increasing. As a result, $\mu^{'}(T)$ is strictly increasing and thus injective. 

Further, note that $\mu^{'}(T)=0$ by construction and $\lim_{T\to\infty} \mu^{'}(T)=+\infty$. In addition, as a sum of continuous functions $\mu^{'}(T)$ is also continuous. This establishes that $\mu^{'}(T)$ is bijective and thus that $\mu(T)$ is invertible on $\mathbb{R}^+$. 

As a result, the Hadamard-L\'evy theorem guarantees that $\mu(T)$ has a unique infinitely differentiable inverse. Given that it has been shown that $\mu^{'}(T)>0$, $\mu(T)$ is strictly increasing and so is its inverse. To show that the inverse is concave, let us define $T\defeq \mu^{-1}(a;y)$. From the chain rule, it follows that
\begin{align*}
(\mu^{-1})^{'}(a;y) = \frac{1}{\mu^{'}(T)}.     
\end{align*}
Using the chain rule once more, one obtains:
\begin{align*}
(\mu^{-1})^{''}(a;y) &=  \frac{d}{dT}\left(\frac{1}{\mu^{'}(T)}\right)\frac{dT}{da}   \\
&= -\frac{\mu^{''}(T)}{(\mu^{'}(T))^2}\frac{dT}{da} 
\end{align*}
Substituting $dT/da = 1/\mu^{'}(T)$, one finally obtains
\begin{align*}
(\mu^{-1})^{''}(a;y) &= -\frac{\mu^{''}(T)}{(\mu^{'}(T))^3}   
\end{align*}
where $T=\mu^{-1}(a;y)$. Since $\mu^{'}(T)>0$, the sign depends on $\mu^{''}(T)$. using the expression for $\mu^{'}(T)$, one can compute the second derivative as follows:
\begin{align*}
\mu^{''}(T) = y\frac{\rho-r}{r(\gamma-1)+\rho}\left(\frac{\rho-r}{\gamma}e^{\frac{\rho-r}{\gamma}T}+re^{-rT}\right)    
\end{align*}
Under the assumption that $\rho>r$, the first term on the right hand side is strictly positive. The two terms under parentheses are exponentials so they are necessarily positive. As a result, it follows that $\mu(T)$ is strictly convex and that its inverse is strictly concave. This concludes the proof of Proposition \ref{prop:HL}. 

\section{Proof of Proposition \ref{prop:Jacobian}}
\label{app:proof_Jacobian}

I first focus on the partial derivatives of $c^\ast(a;y)$ with respect to $a$. Using the chain rule, I get:
\begin{align}
    \frac{\p c^\ast(a;y)}{\p a} &= ye^{\rho\frac{h(a;y)}{\gamma}}\frac{\rho}{\gamma} \left ( -\frac{1}{y} - \frac{\gamma}{\rho} \frac{\partial}{\partial a} W_{-1}(f(a;y)) \right ) \nonumber \\
    &= ye^{\rho\frac{h(a;y)}{\gamma}}\frac{\rho}{\gamma} \left ( -\frac{1}{y} - \frac{\gamma}{\rho} \frac{W_{-1}(f(a;y))}{f(a;y)(1+W_{-1}(f(a;y)))} \frac{\partial}{\partial a} f(a;y) \right) \nonumber \\
    &= ye^{\rho\frac{h(a;y)}{\gamma}}\frac{\rho}{\gamma} \left ( -\frac{1}{y} - \frac{\gamma}{\rho} \frac{W_{-1}(f(a;y))}{f(a;y)(1+W_{-1}(f(a;y)))} f(a;y) \frac{-\rho}{\gamma y} \right) \nonumber \\
    &= ye^{\rho\frac{h(a;y)}{\gamma}}\frac{\rho}{\gamma y} \left ( -1 + \frac{W_{-1}(f(a;y))}{1+W_{-1}(f(a;y))}\right) \nonumber \\
    &= ye^{\rho\frac{h(a;y)}{\gamma}}\frac{\rho}{\gamma y} \left ( -1 + \frac{W_{-1}(f(a;y))}{1+W_{-1}(f(a;y))}\right) \nonumber \\
    &= -e^{\rho\frac{h(a;y)}{\gamma}}\frac{\rho}{\gamma} \frac{1}{1+W_{-1}(f(a;y))}
    \label{eq:expMPC}
\end{align}
where I have used the properties of the Lambert W function to compute its derivative \textemdash see \cite{Mezo2022lambert}, namely that $dW_{-1}(x)/dx = W_{-1}(x)/(x(1+W_{-1}(x)))$. Using the fact that the Lambert W function is such that $W_{-1}(x)e^{W_{-1}(x)} = x$, which in turn implies that $e^{-W_{-1}(x)}=W_{-1}(x)/x$, one can replace $\rho h(a;y)/\gamma$ with $-\frac{\rho a}{\gamma y} - (1+W_{-1}(f(a;y)))$ and obtain
\begin{align}
    \frac{\p c^\ast(a;y)}{\p a} &= -e^{-\frac{\rho a}{\gamma y} - (1+W_{-1}(f(a;y)))} \frac{\rho}{\gamma} \frac{1}{1+W_{-1}(f(a;y))} \nonumber \\
    &= -e^{-\frac{\rho a}{\gamma y} - 1} e^{-W_{-1}(f(a;y))} \frac{\rho}{\gamma} \frac{1}{1+W_{-1}(f(a;y))} \nonumber \\
    &= -e^{-\frac{\rho a}{\gamma y} - 1} \frac{W_{-1}(f(a;y))}{f(a;y)} \frac{\rho}{\gamma} \frac{1}{1+W_{-1}(f(a;y))} \nonumber \\
    &= \frac{\rho}{\gamma}\frac{W_{-1}(f(a;y))}{1+W_{-1}(f(a;y))}
\end{align}
By construction, given that $f(a;y)) \in (-1/e,0)$ one can guarantee that $w=W_{-1}(f(a;y))<1$. As a result, the fraction $w/(1+w)$ is strictly positive given that $1+w<0$. This establishes the first part of Corollary \ref{cor:Jacobian}.

Without loss of generality and for ease of exposition, I now fix $a$ and focus on $c(y)\defeq c^\ast(a;y)$. The same goes for functions $f$ and $h$ defined in Theorem \ref{thm:implicit}. Using the product and chain rules, the derivative of $c(y)$ with respect to $y$ is given by
\begin{align}
\nonumber
c'(y) &= e^{\frac{\rho}{\gamma}h(y)} + y\frac{\rho}{\gamma}\cdot h'(y) e^{\frac{\rho}{\gamma}h(y)}\\
&= e^{\frac{\rho}{\gamma}h(y)}\left[1+\frac{\rho}{\gamma}yh'(y)\right]
\label{eq:deriv_cy}
\end{align}
Let me focus first on the term in brackets in positive. From theorem \ref{thm:explicit}, we know that the function $h$ is given by:
\begin{align}
\label{eq:expression_hay}
h(y) = -\frac{a}{y} - \frac{\gamma}{\rho}\left[1+W_{-1}\left(-e^{-\frac{\rho a}{\gamma y}-1}\right)\right]    
\end{align}
where I have used the expression for $f(y)$ directly into the argument of the Lambert W function. From this, I can compute the derivative of $h$ as follows applying the chain rule a number of times:
\begin{align}
\nonumber
h'(y) &=  \frac{a}{y^2} + \frac{\gamma}{\rho}\frac{\rho a}{\gamma y^2}W^{'}_{-1}\left(-e^{-\frac{\rho a}{\gamma y}-1}\right)e^{-\frac{\rho a}{\gamma y}-1}\\
\nonumber
&= \frac{a}{y^2}\left[1-\frac{1}{e^{-\frac{\rho a}{\gamma y}-1}}\frac{W_{-1}\left(-e^{-\frac{\rho a}{\gamma y}-1}\right)}{1+W_{-1}\left(-e^{-\frac{\rho a}{\gamma y}-1}\right)}e^{-\frac{\rho a}{\gamma y}-1}\right]\\
&=\frac{a}{y^2}\left[1-\frac{W_{-1}\left(-e^{-\frac{\rho a}{\gamma y}-1}\right)}{1+W_{-1}\left(-e^{-\frac{\rho a}{\gamma y}-1}\right)}\right],
\label{eq:hprim_y}
\end{align}
where I have once again used the properties of the Lambert W function to compute its derivative. Simplifying equation \eqref{eq:hprim_y} further, we get
\begin{align*}
h'(y) = \frac{a}{y^2}\frac{1}{1+W_{-1}\left(-e^{-\frac{\rho a}{\gamma y}-1}\right)}
\end{align*}
It follows that the second element of the Jacobian is given by:
\begin{align*}
\frac{\p c^\ast(a;y)}{\p y} = e^{\frac{\rho}{\gamma}h(y)}\left[1+\frac{\rho a}{\gamma y}\frac{1}{1+w}\right].    
\end{align*}
As before, one can use the fact that $e^{-W_{-1}(x)}=W_{-1}(x)/x$ to obtain:
\begin{align*}
\frac{\p c^\ast(a;y)}{\p y} = -W_{-1}(f(a;y))\left[1+\frac{\rho a}{\gamma y}\frac{1}{1+W_{-1}(f(a;y))}\right]    
\end{align*}

Given the fact that $ -W_{-1}(f(a;y))>1$, the consumption function is increasing in $y$ if and only if the term in brackets is strictly positive. This is not obvious given that it is a sum of two terms with opposite signs \textemdash remember that $1+w<0$ so that the second term in brackets in strictly negative. Therefore, the objective is to prove that
\begin{align*}
1+\frac{\rho}{\gamma}y\frac{a}{y^2}\frac{1}{1+W_{-1}\left(-e^{-\frac{\rho a}{\gamma y}-1}\right)} = 1+\frac{\rho a}{\gamma y}\frac{1}{1+W_{-1}\left(-e^{-\frac{\rho a}{\gamma y}-1}\right)}>0    
\end{align*}
As before, I let $w=W_{-1}\left(-e^{-\frac{\rho a}{\gamma y}-1}\right)$. Equivalently, the consumption function is increasing in $y$ if
\begin{align}
\frac{\rho a}{\gamma y}<-(1+w).    
\label{eq:condition_cy_pos}
\end{align}
By construction, the Lambert W function satisfies the following equation:
\[
w e^w = f(y) = -\exp\left(-\frac{\rho a}{\gamma y} - 1\right).
\]
Taking out the exponential term on the right, we get:
\begin{align*}
-w e^w &= \frac{\exp\left(-\frac{\rho a}{\gamma y}\right)}{e}\\
\Rightarrow\quad -we^{w+1} &= \exp\left(-\frac{\rho a}{\gamma y}\right). 
\end{align*}
For convenience, let us define $u=-w>1$. Taking logs on both sides of the equation, we obtain:
\begin{align*}
\log(u) + 1-u = -\frac{\rho a}{\gamma y}    
\end{align*}
Combining with equation \eqref{eq:condition_cy_pos}, we get that the consumption function is increasing in $y$ if and only if
\begin{align*}
u-1-\log(u) <-1+u\quad \Leftrightarrow\quad \log(u) > 0    
\end{align*}
which in turns is equivalent to $u>1$ and thus $w<-1$. We have established that this holds already, therefore we can conclude that $c'(y)>0$. This proves the second part of Corollary \ref{cor:Jacobian}. 

\section{Proof of Proposition \ref{prop:Hessian}}
\label{app:proof_Hessian}

I start with the first element of $\mathbf{H}_c$. From Proposition \ref{prop:Jacobian}, I have
\begin{align*}
\frac{\p c^\ast(a;y)}{\p a} = \frac{\rho}{\gamma}\frac{w}{1+w}    
\end{align*}
Differentiating with respect to $a$ using the chain rule, I obtain: 
\begin{align*}
\frac{\p^2 c^\ast(a;y)}{\p a^2}  = \frac{\rho}{\gamma}\frac{w'}{(1+w)^2}   
\end{align*}
where I have defined $w'\defeq \p W_{-1}(f(a;y))/\p a$. Using once more the chain rule, that derivative can be shown to be equal to:
\begin{align*}
w'&=-\frac{\p }{\p a}W_{-1}(f(a;y))\\
&= \frac{\p }{\p a}f(a;y)W_{-1}^{'}(f(a;y))\\
&= \frac{\p }{\p a}f(a;y)\frac{w}{f(a;y)(1+w)}.
\end{align*}
Note further that one can write
\begin{align*}
\frac{\p }{\p a}f(a;y) = \frac{\rho}{\gamma y}e^{-\frac{\rho a}{\gamma y}-1}.    
\end{align*}
As a result, it follows that 
\begin{align*}
w' &= \frac{\rho}{\gamma y}e^{-\frac{\rho a}{\gamma y}-1}\frac{w}{-e^{-\frac{\rho a}{\gamma y}-1}(1+w)}\\
&= -\frac{\rho}{\gamma y}\frac{w}{1+w}.
\end{align*}
Putting it all together, I obtain:
\begin{align*}
\frac{\p^2 c^\ast(a;y)}{\p a^2}  &= \frac{\rho}{\gamma}\frac{w'}{(1+w)^2}\\
&= -\frac{\rho^2}{\gamma^2 y}\frac{w}{(1+w)^3}.
\end{align*}
Given that $1+w<0$, it follows that $(1+w)^3<0$ and thus that $-w/(1+w)^3<0$. As a result, given the fact that $a>0$, the first element of $\mathbf{H}_c$ is strictly negative. as the second fraction is a quotient of strictly negative values. This proves the claim in Corollary \ref{cor:Hessian} for the first diagonal term.

I now move on to the cross-derivative which defines the off-diagonal terms of $\mathbf{H}_c$. In order to compute the cross-derivative, one can start from
\begin{align*}
\frac{\p}{\p y}c^\ast(a;y) &= e^{\frac{\rho}{\gamma}h(a;y)}\left[1+\frac{\rho a}{\gamma y}\frac{1}{1+w}\right]
\end{align*}
Applying the chain rule, one gets:
\begin{align*}
\frac{\p^2 }{\p a\p y}c^\ast(a;y) &= \left[1+\frac{\rho a}{\gamma y}\frac{1}{1+w}\right]\frac{\rho}{\gamma}\frac{\p }{\p a}h(a;y)e^{\frac{\rho}{\gamma}h(a;y)}    \\
&+ e^{\frac{\rho}{\gamma}h(y)}\left[\frac{\rho}{\gamma y}\frac{1+w-aw'}{(1+w)^2}\right],
\end{align*}
where $w'$ is defined as before. Regrouping terms, I obtain:
\begin{align*}
\frac{\p^2 }{\p a\p y}c^\ast(a;y) &= e^{\frac{\rho}{\gamma}h(a;y)}\left[\left(1+\frac{\rho a}{\gamma y}\frac{1}{1+w}\right)\frac{\rho}{\gamma}\frac{\p }{\p a}h(a;y) + \frac{\rho}{\gamma y}\frac{1}{1+w} - \frac{\rho a}{\gamma y} \frac{w'}{(1+w)^2} \right].
\end{align*}
Focusing first on the partial derivative of the function $h$ defined in equation \eqref{eq:expression_hay}, I get:
\begin{align*}
\frac{\p }{\p a}h(a;y) &= -\frac{1}{y} -\frac{\gamma}{\rho}w'\\
&= -\frac{1}{y}+\frac{\gamma}{\rho}\frac{\rho w}{\gamma y(1+w)}\\
&= -\frac{1}{y(1+w)}.
\end{align*}
Plugging this into the expression for the cross-derivative, I obtain:
\begin{align*}
&e^{\frac{\rho}{\gamma}h(a;y)}\left[-\frac{\rho}{\gamma y (1+w)} - \frac{\rho^2}{\gamma^2y^2}\frac{a}{(1+w)^2}+ \frac{\rho}{\gamma y}\frac{1}{1+w} - \frac{\rho a}{\gamma y} \frac{w'}{(1+w)^2} \right]\\
=\ & e^{\frac{\rho}{\gamma}h(a;y)}\left[-\frac{\rho}{\gamma y (1+w)} - \frac{\rho^2}{\gamma^2y^2}\frac{a}{(1+w)^2}+ \frac{\rho}{\gamma y}\frac{1}{1+w} + \frac{\rho^2 }{\gamma^2 y^2} \frac{aw}{(1+w)^3}
\right]\\
=\ & \frac{e^{\frac{\rho}{\gamma}h(a;y)}}{\gamma^2 y^2 (1+w)^3}\left[-\rho\gamma y (1+w)^2-\rho^2 a(1+w)+\rho\gamma y(1+w)^2+\rho^2aw\right]\\
=\ & -\frac{e^{\frac{\rho}{\gamma}h(a;y)}}{\gamma^2 y^2 (1+w)^3}a\rho^2.
\end{align*}
Using the fact that the exponential in the numerator simplifies to $-w$, I finally obtain
\begin{align*}
\frac{\p^2 c^\ast(a;y)}{\p a \p y}  =\frac{a\rho^2}{\gamma^2 y^2}\frac{w}{(1+w)^3}. 
\end{align*}
Now using the fact that $(1+w)^3<0$ given the properties of the Lambert W function, it follows that 
\begin{align*}
\frac{\p^2 }{\p a\p y}c^\ast(a;y)>0    
\end{align*}
as the second fraction is a quotient of strictly negative values. This proves the claim in Corollary \ref{cor:Hessian} about the off-diagonal terms.

The last step is to prove that the consumption function is strictly concave in $y$: $c''(y)<0$. Taking the derivative of \eqref{eq:deriv_cy} to compute $c^{''}(y)$, I get
\begin{align}
\nonumber
c^{''}(y) &= \frac{\rho}{\gamma}h'(y)e^{\frac{\rho}{\gamma}h(y)}\left[1+\frac{\rho}{\gamma}yh'(y)\right] + e^{\frac{\rho}{\gamma}h(y)}\frac{\rho}{\gamma}\left[h'(y)+yh^{''}(y)\right]\\
\nonumber
&= \frac{\rho}{\gamma}e^{\frac{\rho}{\gamma}h(y)}\left[h'(y)\left(1+\frac{\rho}{\gamma}yh'(y)\right)+h'(y)+yh^{''}(y)\right]\\
&= \frac{\rho}{\gamma}e^{\frac{\rho}{\gamma}h(y)}\left[2h'(y)+\frac{\rho}{\gamma}y(h'(y))^2+yh^{''}(y)\right]
\label{eq:cprimprim}
\end{align}

Using the quotient and chain rules, we get
\begin{align*}
h''(y) = -\frac{a}{y^4 (1 + w)^2
} \left[ 2 y (1 + w) + \frac{w \rho a}{\gamma (1 + w)} \right].
\end{align*}
Using this, we can rewrite the term withing brackets in equation \eqref{eq:cprimprim} as follows:
\begin{align*}
&\frac{2a}{y^2(1+w)} + \frac{\rho y}{\gamma}\frac{a^2}{y^4(1+w)^2} -\frac{a y}{y^4 (1 + w)^2
} \left[ 2 y (1 + w) + \frac{w \rho a}{\gamma (1 + w)} \right]\\
=& \frac{2a}{y^2(1+w)} + \frac{\rho }{\gamma}\frac{a^2}{y^3(1+w)^2} - \frac{2a}{y^2(1+w)} - \frac{\rho a^2 w}{\gamma y^3(1+w)^3}\\
=& \frac{\rho a^2}{\gamma y^3(1+w)^2}\left(1-\frac{w}{1+w}\right) \\
=& \frac{\rho a^2}{\gamma y^3(1+w)^3}
\end{align*}
Using once again the fact that the exponential on the right simplifies to $-w$, I finally obtain
\begin{align*}
\frac{\p^2 c^\ast(a;y) }{\p y^2} = -\frac{ \rho^2a^2}{\gamma^2 y^3} \frac{w}{(1+w)^3},
\end{align*}
which completes the proof of Proposition \ref{prop:Hessian}. Regarding the sign of this second order derivative, it is strictly negative just as for the second order derivative with respect to $a$. This completes the proof of Corollary \ref{cor:Hessian}. 

\end{document}